\pdfoutput=1
\documentclass[10pt,journal,final,twocolumn]{IEEEtran}
\usepackage{bm}
\usepackage{mathrsfs}
\usepackage{textcomp}
\usepackage{epsfig}
\usepackage{indentfirst}
\usepackage{amsmath}
\usepackage{amssymb}
\usepackage{epstopdf}
\graphicspath{{../}}
\usepackage{array}
\usepackage{multirow}
\usepackage{graphicx}
\usepackage{subfigure}

\usepackage{amsthm}
\newtheorem{theorem}{Theorem}
\newtheorem{lemma}[theorem]{Lemma}

\theoremstyle{definition}

\theoremstyle{remark}

\usepackage{cite}

\begin{document}

%
\title{
Analysis of RF Energy Harvesting in Uplink-NOMA IoT-based Network
}

\author{
   \IEEEauthorblockN{Zhou Ni\IEEEauthorrefmark{1}, Ziru Chen\IEEEauthorrefmark{1}, Qinbo Zhang\IEEEauthorrefmark{1}, Chi Zhou\IEEEauthorrefmark{1}}

     \IEEEauthorblockA{\IEEEauthorrefmark{1}Department of Electrical and Computer Engineering, Illinois Institute of Technology, Chicago, USA }
     

 \IEEEauthorblockA{Emails: \{zni1,zchen71,qzhang63\}@hawk.iit.edu}, zhou@iit.edu
}

\maketitle

\begin{abstract}
Internet of Things (IoT) systems in general consist of a lot of devices with massive connectivity. Those devices are usually constrained with limited energy supply and can only operate at low power and low rate. In this paper, we investigate a cellular-based IoT system combined with energy harvesting and NOMA. We consider all base stations (BS) and IoT devices follow the Poisson Point Process (PPP) distribution in a given area. The unit time slot is divided into two phases, energy harvesting phase in downlink (DL) and data transmission phase in uplink (uplink). That is, IoT devices will first harvest energy from all BS transmissions and then use the harvested energy to do the NOMA information transmission. We define an energy harvesting circle within which all IoT devices can harvest enough energy for NOMA transmission. The design objective is to maximize the total throughput in uplink within the circle by varying the duration $T$ of energy harvesting phase. In our work, we also consider the inter-cell interference in the throughput calculation. The analysis of Probability Mass Function (PMF) for IoT devices in the energy harvesting circle is also compared with simulation results. It is shown that the BS density needs to be carefully set so that the IoT devices in the energy harvesting circle receive relatively smaller interference and energy circles overlap only with a small probability. Our simulations show that there exists an optimal $T$ to achieve the maximum throughput. When the BSs are densely deployed consequently the total throughput will decrease because of the interference.

\end{abstract}


\section{Introduction}
With the development of 5G technology, it is possible to transmit more data and support more users (UEs) in nowadays wireless communication. Non-orthogonal Multiple Access (NOMA) is a key technology in 5G and has been extensively studied. Allowing serval data sources to be transmitted at the same frequency band and the same time slot is one of the main advantages of NOMA. 
With the growing demand for communication spectrum, NOMA technology has been widely considered as a powerful way to improve the spectrum efficiency in the future mobile communication system. 
Since it was proposed in 2013, numbers of NOMA technologies have been come up for the mobile communication network. It can be roughly divided into downlink (DL) NOMA and uplink (UL) NOMA. 

DL NOMA has been widely studied in recent years.
A paring UEs are usually considered in a basic NOMA system. For the transmitter, it uses non-orthogonal transmission which combines all the NOMA UEs information in the DL signal. At the receiver side, Successive Interference Cancellation (SIC) is applied when UEs receive the signal. SIC relies on decoding and subtracting the signals in sequence until it reaches its desired signal~\cite{manglayev2017noma}. Therefore, the decoding order must match with the UE index in the cancellation sequence to get the right information. Zekun. et al. proposed a framework on
the NOMA system coverage and average UE achievable data rate in a stochastic geometry-based NOMA system~\cite{zhang2016stochastic}.
in~\cite{zhang2016stochastic} where a multi-cell DL cellular network is considered in its system model. In each cell,
there are two UEs connected to the BS. All BSs are set to transmit a fixed power $P$ . In addition, the inter-cell interference is considered in this paper and can be calculated by using the Laplace transform. According to the result of this paper, the NOMA has a negative influence on the performance of UE SIR. However, NOMA improves the overall system throughput. Some power allocation schemes for DL NOMA are proposed in~\cite{wang2016power}, it investigates power allocation for a DL NOMA system. In their system model, single BS and two UEs are considered as well. There are two closed-form optimal power allocation schemes derived by using Karush-Kuhn-Tucher (KKT) conditions. Moreover, a Poisson cellular network applying NOMA in the DL is studied in~\cite{ali2018analyzing}. In the system model, all the NOMA UEs are distributed randomly in the largest disk centered at the BS in the cell.

Previous works for NOMA mainly concentrates on DL NOMA. The investigation for UL NOMA is relatively fewer. For UL NOMA, the eNB has to receive different arrived power from all NOMA UEs. How to obtain power is the main issue. In~\cite{zhang2016uplink}, it comes up a power back-off scheme and investigates the performance of outage probability for two NOMA UEs in a single cell transmission case. In addition, user pairing in NOMA is studied in many predefined power allocation schemes. A set of UEs are divided by BS into disjunct pairs and sent the available power to these pairs. Different scenarios, such as SAMS, MBASS, are considered in~\cite{sedaghat2018user}. In~\cite{ali2018downlink}, the authors proposed three models for the UEs in Poisson Cluster process (PCP) system which gives the idea of the range. The UEs in this work is considered to be distributed in a certain range in the three systems and analyzed. In addition, all the NOMA UEs are connected with their nearest BS.

Although NOMA has many advantages for 5G, challenges and obstacles exist when NOMA is applied into practice. For example, when the number of UE in a cell becomes very large, the complexity of decoding will increase as well. We know that each of the UEs in the NOMA cluster needs to decode all the information they received even one having worst channel condition. Therefore, when the system has a large number of UEs, the complexity and power consumption will be significantly higher. Moreover, all other UEs decoding information will be erroneous if any error occurs in single UE due to SIC. This constrains the number of UEs to be served in NOMA at the receiver. Fortunately, more and more new techniques are applied to solve those problems. NOMA is combined with many other systems or mathematics tools, such as machine learning, nowadays. 

In IoT networks, wireless charging has been proposed recently as a promising solution to supply power to a large number of users\cite{8726578,chen2017optimal}.
In~\cite{diamantoulakis2016wireless}, it combines energy harvesting with a NOMA system. The energy harvesting process is studied because IoT devices and nodes are inconvenient to be removed and charged in wireless IoT networks. For example, some nodes are put underground to collect information about the soil or waves. Therefore, energy harvesting technology gives the self-sustainable ability to these IoT nodes. With the development of these new technologies, we will have a more reliable and powerful wireless communication system in the foreseeable future. 
Wireless powered ultra-dense networks (UDNs) has shown its potential to be a good candidate of wireless RF energy harvesting. In~\cite{li2018hybrid,he2013recursive,chen2017energy}, the RF energy have been harvested from the environment in the UDNs, in which the distance between the BSs and users are significantly reduced.
UDNs could not only improve the link qualities but also increase the harvested RF energy levels. However, increasing intercell interference may influence network performance. Hence, the density of BSs needs to be carefully designed.


In addition, Stochastic geometry (SG) has been widely adopted to characterize the random deployment of BSs and wireless users for network performance analysis, which can be employed to quantify the co-channel interference~\cite{zhang2014stochastic} and also be incorporated with random channel effects such as fading and shadowing~\cite{keeler2013sinr}.

In this paper, we modeled an RF-energy harvesting uplink-NOMA IoT-based Network by using Poisson point process. The main contribution is in three-fold. First, we propose an energy-efficient NOMA system. Specifically, UEs in a ``selection range'' will transmit information to BS by using uplink NOMA. ``Selection range'' is defined as a circle in which UEs can always harvest enough energy~\footnote{If UEs are not located in this ``selection range'', we assume they will use other sub-channel to transmit (OFDMA). So in this paper, we only focus on the users located in the ``selection range''.}. Second, we
analyze the user selection scheme and average system throughput numerically. Third, extensive
simulations are conducted, in which we can find some useful insight into system design.

The remainder of this paper is organized as follows. In Section~\ref{System.model}, the variables and assumptions governing the model are listed. In Section~\ref{ITA}, user selection strategy has been analyzed. In Section~\ref{Simulation}, simulation results have been presented. Finally, Section~\ref{Conclusion} briefly concludes the article and proposes several future research direction. 

\begin{figure}[t]
\centering
\includegraphics[width=0.45\textwidth]{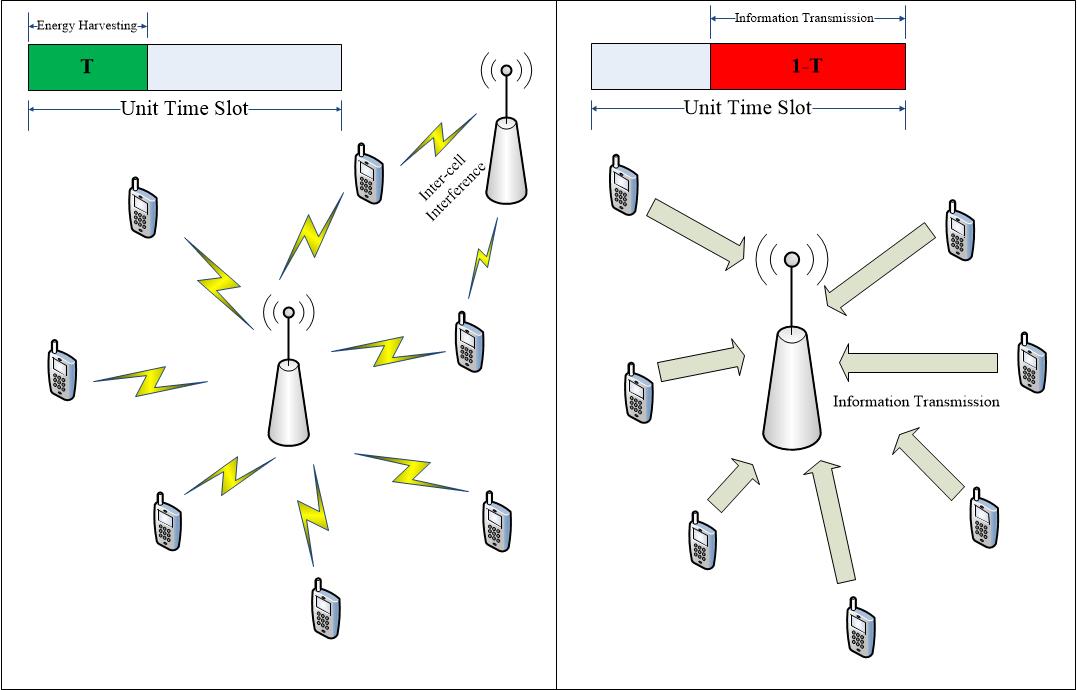}
\caption{ Illustration of the system setup and the two sub-slots (Energy harvest, uplink). }
\label{Fig}
\end{figure}

\section{System Model}\label{System.model}
An ultra-dense cellular network is considered where BSs and IoT devices are distributed following two independent homogeneous Poisson point process (PPP).
The locations of BSs and IoT devices are modeled on the Euclidean plane $\mathbb{R}^2$ as $\Phi_B={b_i, i=1,2,...}$ and $\Phi_U={u_i,i=1,2,....}$ with density $\lambda_B$, and $\lambda_U$, respectively.
Without loss of generality, each IoT device associates the nearest BS, in other words, the coverage area of each BS comprises a Voronoi cell.
IoT devices harvest RF energy from surrounded active BSs and use those energies to do downlink information reception and uplink information transmission. 
BSs and IoT devices are both equipped with one single antenna and each of them operates over the same frequency band.

 As shown in Fig.\ref{Fig}, for each transmission, each time slot has been separated into two sub-slots. For the first sub-slot, all UEs located inside of the selection range will harvest energy from all BSs in the whole system~\footnote{'To simplify the presentation, all `UEs' used in the following paper is represent the UEs located in the selection range'}. Then, they use the energy that they harvested to do the uplink NOMA information transmission in the transmission sub-slot. Especially, for the energy harvesting part, a single UE can harvest energy from all BSs in the whole system and will be selected by its associated BS to apply uplink NOMA transmission scheme if its harvested sufficient energy in energy harvesting sub-slot.
 Specifically, we consider the power domain NOMA system and apply SIC to cancel all intra-cell interference. By applying SIC, when BS receives the signal from all UEs in its cell, the BS first detects and decodes the signal which is the strongest one and extracts this signal from the combined received signal and does the same process until all signals are detected and decoded. We assume the SIC can be a success for all the selected users and all users always have packages to transmit.

Both small-scale fading and large-scale fading have been considered in our paper. Specifically, rayleigh fading is considered when it comes to small scale fading, and each sub-slot has different fading factors $h_x, g_x\sim \exp(1)$. $h_x$ and $g_x$ are denotes the small scale fading for energy harvesting sub-slot and UL transmission sub-slot respectively.
Besides, large-scale fading is considered decays at the rate of $r^{-\alpha}$, where $r$ is the distance between UEs and BSs, and $\alpha>2$ is the path loss exponent.

When we consider a typical UE located at the origin, the received RF energy of a typical UE at energy harvesting sub time-slot is thus
\begin{align}\label{sys.eh_d}
    E_{\mathrm{HD,j}} = T a P_{bs} h_{1,1}r_1^{-\alpha}+\sum_{b_i\in\Phi_{B}^\prime} T a{P_{bs} h_{i,j}r_{i}^{-\alpha}},
\end{align}
where $T$ is the energy harvesting sub-slot fraction, $a \in (0,1)$ is the energy conversion efficiency from RF to DC, $h_i$ denotes the channel gain between the typical IoT device to any BS $b_i \in \Phi_B^\prime$ where $b_i$ is the $i$-th nearest BS to the typical IoT device and $\Phi_B^\prime$ represents BSs beyond the nearest one who are operating in the energy harvesting or downlink transmission sub-frame. $r_i$ is the distance between the typical IoT device and its $i$-th nearest BS.

\section{User Selection and System Throughput Analysis}\label{ITA}
In this section, we study the UEs selection scheme and system throughput. For a UE to successfully selected by its associated BS to operate uplink NOMA, it has to harvest sufficient energy at energy harvesting sub-slot to operate the uplink transmission. If conditioned on the distance between the UE to its associated BS, we define an energy harvesting range in which UEs can harvest sufficient RF energy with probability $\beta$, $\beta$ is the power coverage requirement. The radius for this range denoted as $r$ is in the following lemma,
\begin{lemma}\label{lm.EK}
If UEs in the selection range can receive enough RF energy to proceed uplink transmission for uplink information transmission sub-slot, they will upload information using NOMA technology. The radius of selection range is given by:
\begin{align}
r = \sqrt{\left(\ln(\beta)+\frac{E_{\mathrm{th}}}{TaP_S}\right)\frac{\alpha-2}{2\pi\lambda_B}},\label{lemma1}
\end{align}
where $E_\mathrm{th}$ is the energy consumption threshold.
\end{lemma}

\begin{proof}
Since the selection of NOMA UEs is based on the energy harvesting status, the energy harvesting sufficient probability condition on the distance between the typical UE to its associated BS $r_1$ is give by,
\begin{align}
 &\hspace{4pt} \mathbb{P}\left(E_\mathrm{H} \geq E_{\mathrm{th}} \bigg |R=r_1 \right) \nonumber
\hspace{-3pt}\\
=
& \mathbb{P} \left(aP_S h_1 r_1^{-\alpha}+\sum_{b_i\in\Phi_B/b_1} aP_S h_i r_i^{-\alpha} \geq \frac{E_\mathrm{th}}{T}\bigg| R=r_1\right), \nonumber \\
\hspace{-3pt} {=} & \mathbb{P}\left(h_1 r_1^{-\alpha}+\mathbb{E}\left[\sum_{b_i\in\Phi_B/b_1}h_i r_i \geq \frac{E_{\mathrm{th}}}{TaP_S} \right]  \bigg | R=r_1 \right), \nonumber \\
 \hspace{-3pt} \overset{(a)}{\approx}&  \mathbb{P}\left(h_1\geq  \frac{E_{\mathrm{th}}r_1^\alpha}{TaP_S}-\frac{2\pi\lambda_B r_1^{2}}{\alpha-2} \bigg| R=r_1\right)\nonumber \\
 \hspace{-3pt} = &\exp\left(-\frac{E_{\mathrm{th}}}{TaP_S}+\frac{2\pi\lambda_Br_1^2}{\alpha-2}\right),\label{EH}
\end{align}
where in equation (a) we can approximate as:
\begin{align} 
&\mathbb{E}_{\Phi_B}\left[\sum_{b_i \in \Phi_{B}/b_{1}}  h_i r_{i}^{-\alpha}  \bigg |R=r_{1}\right] \nonumber\\
{=} &\mathbb{E}_{\Phi_B}\left[\sum_{b_i \in \Phi_{B}/b_{1}}  r_{i}^{-\alpha} \bigg| R=r_{1}\right], \nonumber \\
 \overset{(b)}{=} &2\pi \lambda_B\int_{r_1}^{\infty}\frac{1}{r^\alpha}rdr, \nonumber \\
 = &\frac{2\pi  \lambda_B r_1^{2-\alpha}}{\alpha-2}, \nonumber
\end{align}
where (b) comes from the fact that Campbell's theorem for summation over PPP~\cite{kishk2017joint}. 

If we put \eqref{EH} equals to $\beta$, and calculate $r_1$, lemma 1 proved.
\end{proof}

After select UEs, we are interested in the total throughput after SIC. 

All UEs use the harvested energy to upload informations into its associated BS by using power $P_u = E_\mathrm{th}/(1-T)$.
If we use $N$ to denotes the number of UEs been selected by typical BS, the data rate for the $m$th UE ($m\leq N$) conditioned on $N$ can be written as:
\begin{align}
    R_m &= (1-T)\log_2\left(1+\frac{S_m}{I_{inter}+\sum_{i=m+1}^N S_i+\sigma}\right)\nonumber\\
    &=(1-T)\log_2\left(\frac{I_{inter}+\sum_{i=m}^N S_i+\sigma}{I_{inter}+\sum_{i=m+1}^N S_i+\sigma}\right),\nonumber
\end{align}
where $S_i = P_u g_i r_i^{-\alpha}$ represent the received signal at the BS side and $I_{inter}$ is the inter-cell interference. Moreover, $\sigma$ represents the gaussian white noise.

The average total throughput for each cell denoted as $R_{tc}$ is given by:
\begin{align}
    R_{tc} &= \mathbb{E}_N\left[(1-T)\frac{\sum_{i=1}^NS_i}{I_{inter}+\sigma}\right]\nonumber\\
    &=\sum_{n=1}^{\infty}\left[(1-T)\frac{\sum_{i=1}^nS_i}{I_{inter}+\sigma}\mathbb{P}(N=n)\right],
\end{align}
where the probability mess function (PMF) of $N$ can be written as:
\begin{align}\label{pn}
    \mathbb{P}(N=n) = \frac{\lambda_u\pi (r)^2}{n!}e^{\lambda_u\pi (r)^2},
\end{align}
where $r$ is given in \eqref{lemma1}.

System average data rate denoted as $R_{ts}$ is given by:
\begin{align}\label{rts}
    R_{ts} = \lambda_B^\prime R_{tc},
\end{align} 
where $\lambda_B^\prime$ is the density of BSs receiving uplink NOMA signals, it can be written as:
\begin{align}
    \lambda_b^\prime = \lambda_b*(1-\mathbb{P}(N=0)).
\end{align}

Therefore, we can calculate the total data rate of the system.

\section{Simulation and Results}\label{Simulation}
The results were obtained using the parameters in Table I.
We modeled the multi-UEs and BSs as Poisson points and calculated the total data rate for the system. In addition, we used 50,000 samples to average every results in this simulation.
\begin{table}[t]\label{sim}
\centering
\caption{Parameters in simulation}
\setlength{\tabcolsep}{9mm}{
\begin{tabular}{l|l}
\hline
\hline
\textbf{Parameters} & \textbf{Values} \\
\hline
 $\lambda_b$ &   $20$ to $40$ BS/$km^2$ \\
 \hline
  $P_S$& $1$J\\
  \hline
   $\alpha$& 4 \\
\hline
   $E_{\mathrm{th}}$& $10^{-4}$J\\
  \hline
   $a$  &$0.5$\\ 
   \hline
   $\beta$&$99\%$\\ 
\hline         
\end{tabular}}
\end{table}

The relationship between $r$ and $\lambda_b$ is shown in Fig.~\ref{LambdaBwithR}. It is easy to find that a higher $T$ will achieve a larger energy
harvesting circle. This is because the system has a long energy harvesting time which more UEs can get enough energy for uplink transmission even those UEs far from their connected BS. Meanwhile, with higher BS density which means the distance between a UE to all BSs are closer. UE also can harvest more energy. The system can support more NOMA UEs and the range becomes larger.
\begin{figure}[hthp]
\centering
\includegraphics[width=0.53\textwidth]{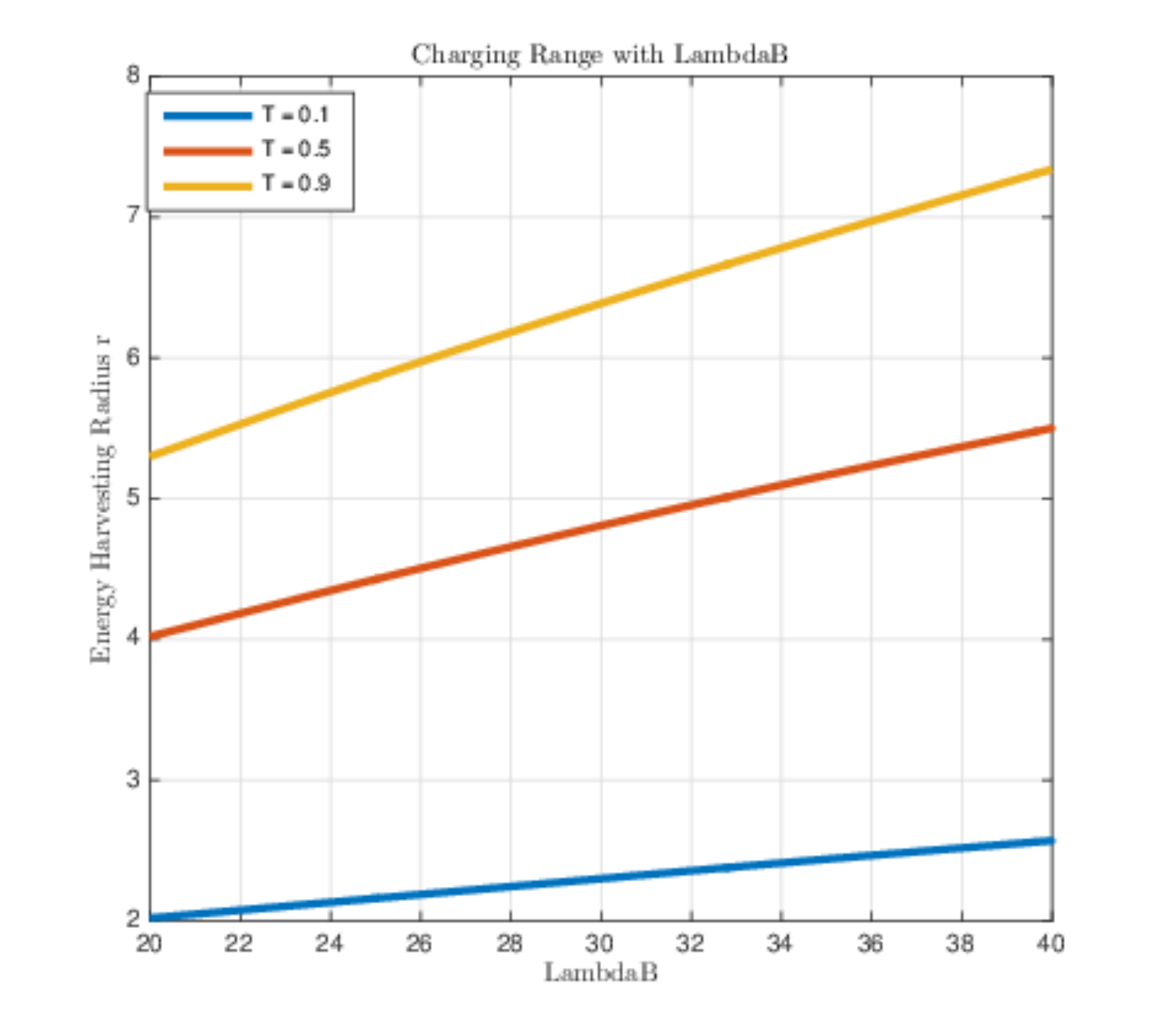}
\caption{Relation between $r$ and $\lambda_b$. }
\label{LambdaBwithR}
\end{figure}

In Fig.\ref{PMFofUE}, we show the accuracy of our equation \eqref{pn} who gives the PMF of how many UEs in the circle. Then we need to compare the analysis result with simulation result to make sure that they can match with each other. In this simulation, we randomly distribute the BSs and UEs following the PPP. Then one typical cell is extracted from all the cells to present other cells. Therefore, we need to consider the UE distribution probability in the typical cell and compare it with mathematical analysis result.
We can imagine that when we have a large energy harvesting range, the overlap probability of these circles will significantly increase. When $\lambda_b$ reaches up to $300 bs/km^2$, the overlap probability is $100\%$, which means most UEs in the same energy harvesting circle but they belong to different cells.  In the analysis, it just calculates how many UEs in this range regardless of whether these UEs are located in the same cell. But in simulation, we just calculate UEs in the same cell. Therefore, the error comes when $\lambda_b$ becomes larger. This is the reason why we chose a certain number of $\lambda_b$ which is $20$ to $40$ in our work.

\begin{figure}[hthp]
\centering
\includegraphics[width=0.53\textwidth]{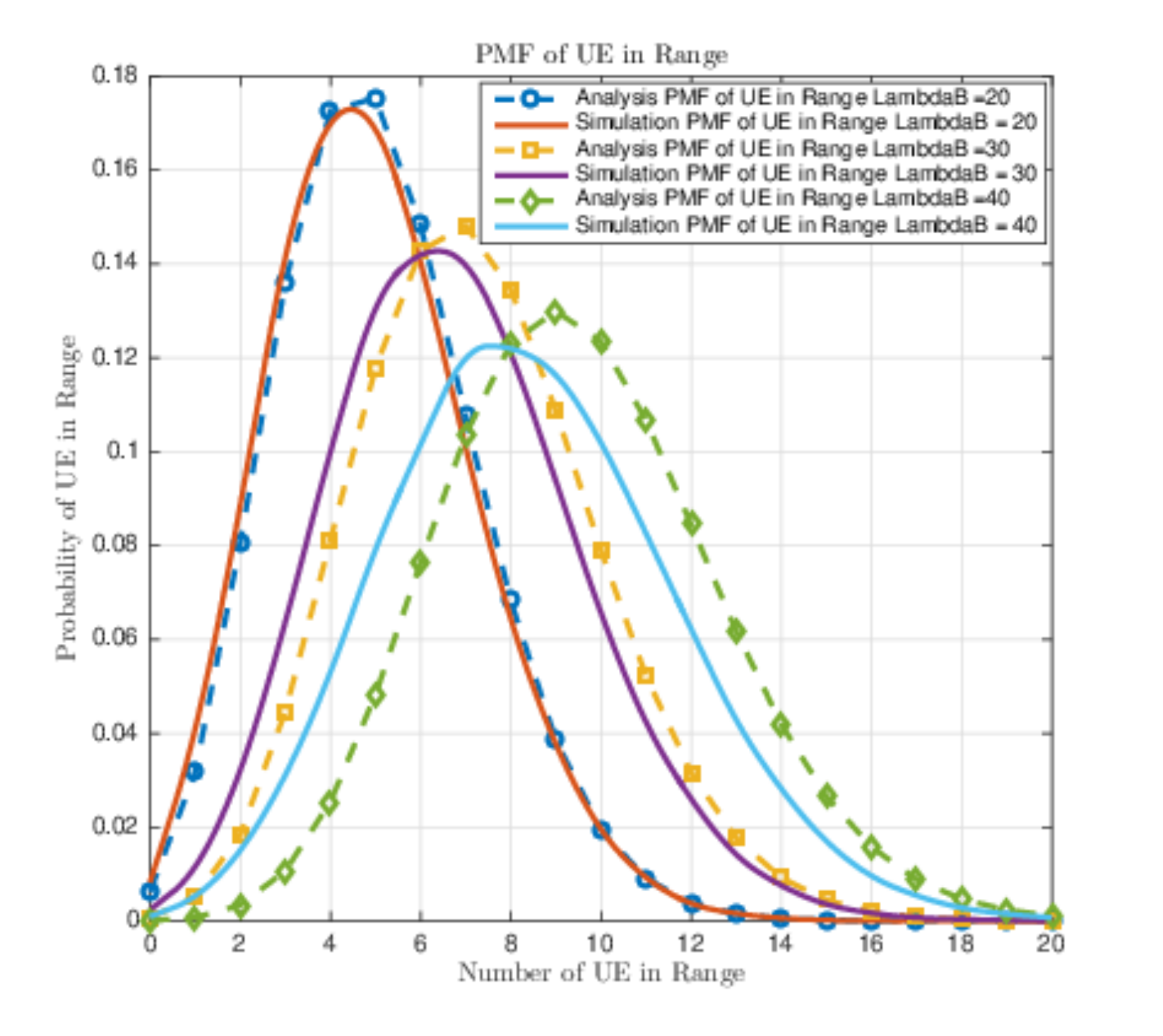}
\caption{PMF of UEs in circle. }
\label{PMFofUE}
\end{figure}

The system throughput with changed $T$ can be calculated according to equation~\eqref{rts}. From this equation, we can find that the total data rate of the whole system is related to the information transmission time and inter-cell interference. In Fig.\ref{throughputWithTau}, we show that with same energy harvesting time $T$, smaller $\lambda_b$ can reach up to a higher system throughput compared with larger BS density. This is because the smaller $\lambda_b$ introduced less uplink interference.  As for the fixed $\lambda_b$ but different $T$, the throughput will increase with a short energy harvesting time and $T = 0.15$ gives the maximum throughput of the system. When the harvesting time becomes larger than this optimal value, the throughput becomes decrease because larger $T$ also achieves a larger $r$ which means high inter-cell interference. The reason behind this is the tradeoff of energy harvesting sub time-slot and transmission sub time-slot.
\begin{figure}[t]
\centering
\includegraphics[width=0.53\textwidth]{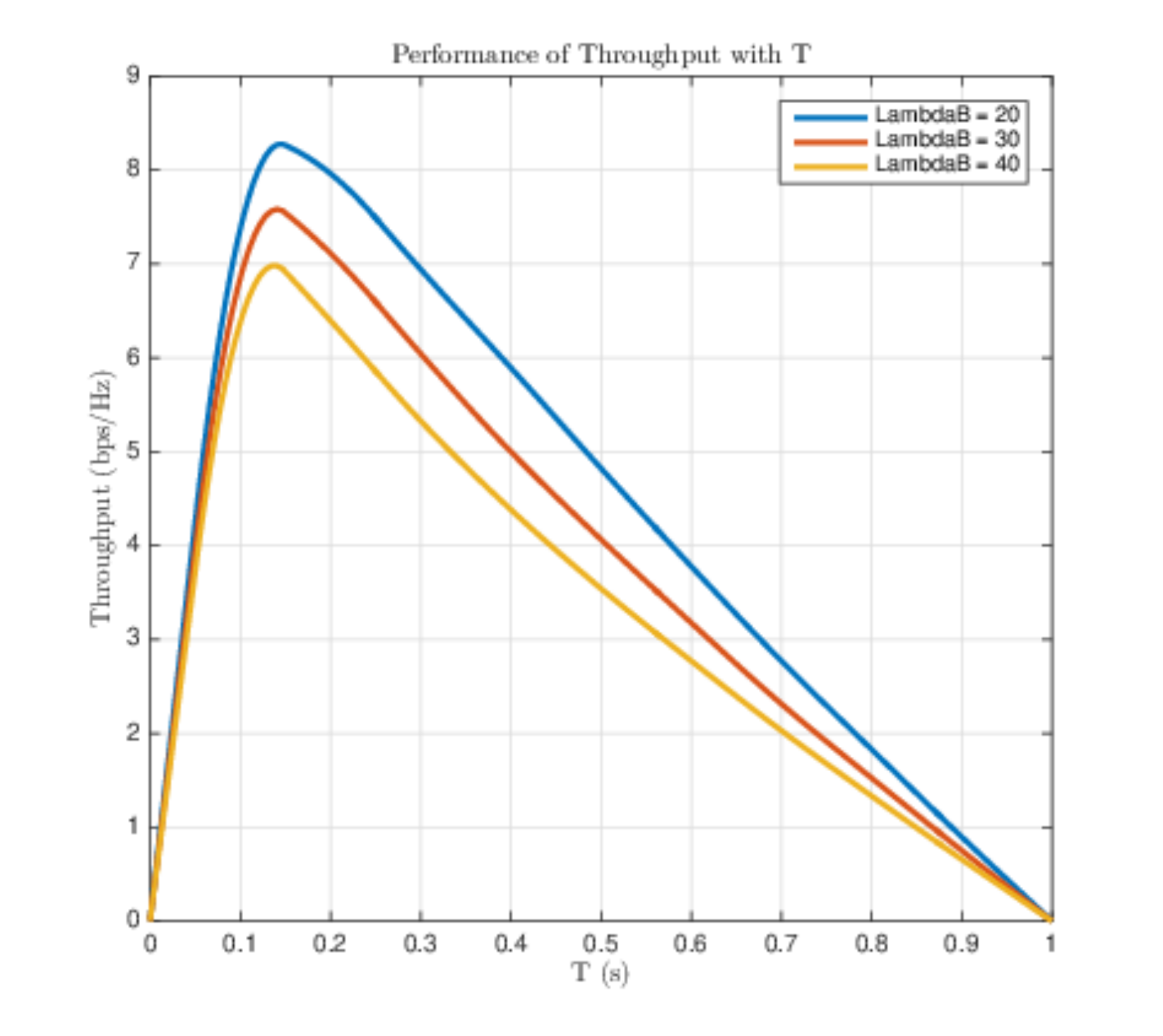}
\caption{System throughput with changed 𝑇. }
\label{throughputWithTau}
\end{figure}

In many traditional NOMA works who just use a single cell to present the whole system. To get optimal system throughout, most of them ignored the inter-cell interference. However, it was unpractical when multi-BSs and UEs in the system. The interference will become too high to be ignored. In Fig.\ref{duibiinterference}, we find out that inter-cell interference seriously affected the system throughout.
\begin{figure}[hthp]
\centering
\includegraphics[width=0.53\textwidth]{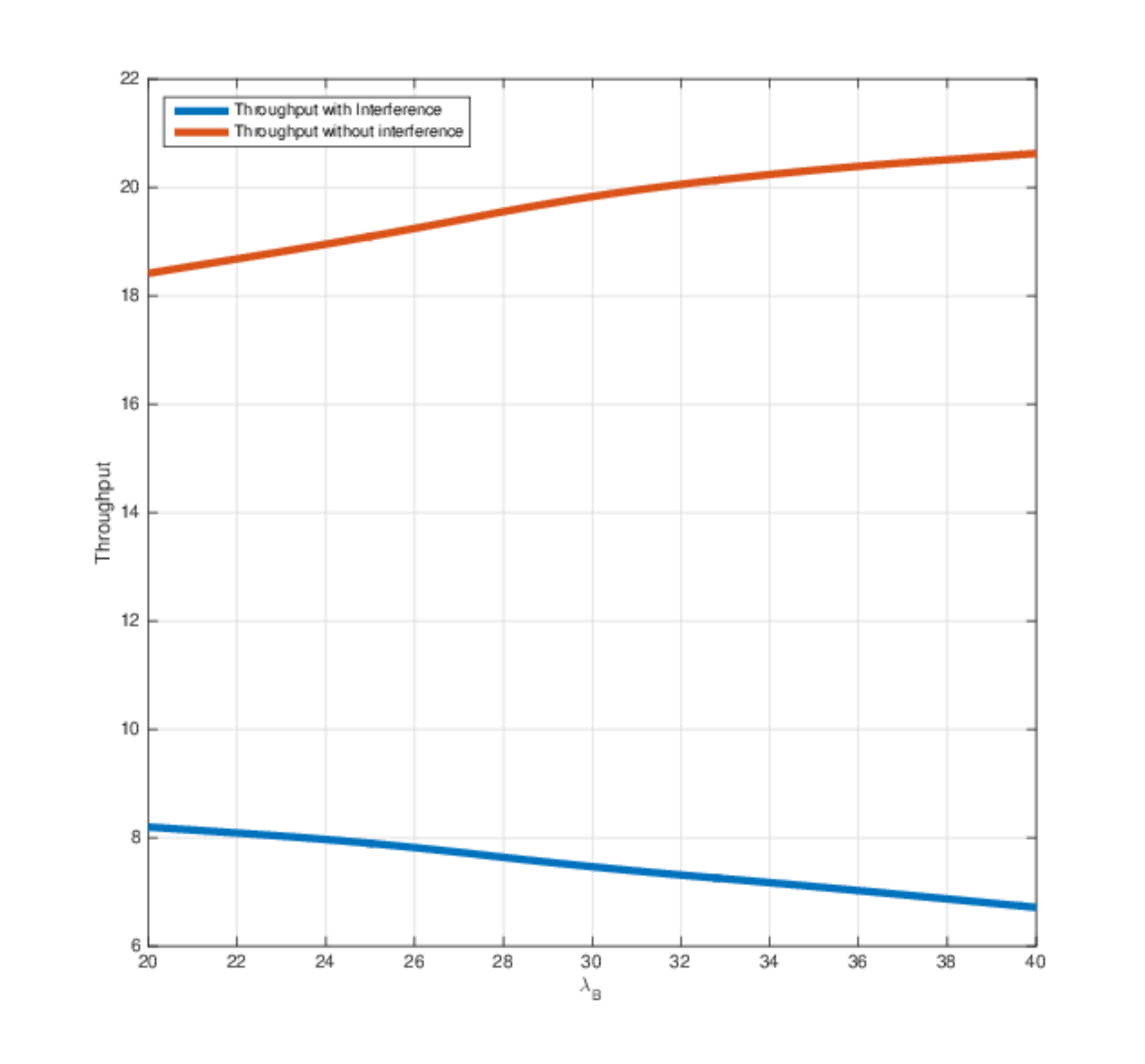}
\caption{Compare throughput with and without inter-cell interference. }
\label{duibiinterference}
\end{figure}

\section{Conclusion}\label{Conclusion}
In this paper, we have provided a new model for uplink NOMA transmission in a PPP system and developed an analytical framework to study uplink NOMA performance of a cellular-based ambient RF energy harvesting network in which IoT devices are solely powered by the downlink cellular transmissions. We proposed a multi-UEs and multi-BSs scenario with inter-cell interference which gives a model for a general communication network. Based on the more realistic network, optimal slot partitioning that maximizes this throughput under different density of BSs have been discussed. Simulation results show that inter-cell interference has a dramatic influence for a stochastic geometry system.

We will extend the work by considering imperfect SIC in our future work.

\vspace{16pt}
\bibliographystyle{IEEEtran}
\bibliography{IEEEfull,new}

\end{document}